\documentclass[3p,11pt]{elsarticle}

\usepackage{amssymb}
\usepackage{latexsym,amsmath,amsthm,subfigure,url,pxfonts}

\journal{}

\newcommand{\eps}{\varepsilon}
\newcommand{\abs}[1]{\left\vert#1\right\vert}

\newcommand{\p}{\partial}
\newcommand{\mB}{\mathbf{B}}

\newcommand{\ma}{\mathbf{a}}

\newcommand{\mx}{\mathbf{x}}
\newcommand{\my}{\mathbf{y}}
\newcommand{\mz}{\mathbf{z}}
\newcommand{\vn}{\boldsymbol{\nu}}
\newcommand{\vt}{\boldsymbol{\theta}}

\newtheorem{thm}{Theorem}[section]

\begin{document}

\begin{frontmatter}



\title{One-step iterative reconstruction of conductivity inclusion via the concept of topological derivative}


\author{Won-Kwang Park}
\ead{parkwk@kookmin.ac.kr}
\address{Department of Mathematics, Kookmin University, Seoul, 136-702, Korea.}

\begin{abstract}
We consider an inverse problem of location identification of small conductivity inhomogeneity inside a conductor via boundary measurements which occurs in the EIT (Electrical Impedance Tomography). For this purpose, we derive topological derivative by applying the asymptotic formula for steady state voltage potentials in the existence of conductivity inclusion of small diameter. Using this derivative, we design only one-step iterative location search algorithm of small conductivity inhomogeneity completely embedded in the homogeneous domain by solving an adjoint problem. Numerical experiments presented for showing the feasibility of proposed algorithm.
\end{abstract}


\end{frontmatter}





\section{Introduction}\label{sec:1}
Electrical Impedance Tomography (EIT), an technique for imaging distribution of the conductivity distribution of conducting objects from surface electrical measurements, is an interesting and important problem which is arising in physics, medical science, material engineering, and so on, all domains highly related with human life. Related works can be found in \cite{AK1,AK2,B2,C,KSS,KSY,SKAW} and reference therein. However, due to the ill-posedness and inherent non-linearity of problem, it remains a challenging research area. In recent research, various remarkable imaging techniques proposed, many of them based on the Newton-type iteration method, refer to \cite{ADIM,KSS,SV} and references therein. For a successful application of such method, a good initial guess that is close to the unknown target is essentially required. If one proceed without a good initial guess, issue such as non-convergence, the occurrence of several minima, and large computational costs may arise. Moreover, iteration method often requires suitable regularization terms that highly depend on the specific problem at hand.

For this purpose, various non-iterative methods are also developed as an alternatives, e.g., a real-time algorithm for finding location of conductivity inhomogeneity \cite{AK1,AK2,KSY,SKAW}, simple pole method \cite{KL}, and end-point location search algorithm of thin conductivity inhomogeneities \cite{ABF1,ABF2,LP}.

The main purpose of this paper is to develop a fast, non-iterative imaging algorithm of small conductivity inclusion embedded in a homogeneous domain from boundary measurement by adopting famous topological derivative concept \cite{CR,DL,SZ}. By applying asymptotic expansion formula in the existence of conductivity inclusion of small diameter \cite{AK1,AK2,AS}, we can easily derive desired derivative and design a fast (only requires one-step iteration procedure by solving an adjoint problem) imaging algorithm. Although, the result via topological derivative does not guarantee complete shape of unknown target, it will be a good initial guess of Newton-type iteration method.

This paper organized as follows. In section \ref{sec:2}, we briefly review basic mathematical model of direct conductivity problem and introduce the asymptotic expansion formula due to the existence of small inclusion. In section \ref{sec:3}, we apply asymptotic expansion formula in order to rigorously and easily derive the topological derivative and develop a one-step iterative imaging algorithm. In the following section \ref{sec:4} numerical experiments are shown. Section \ref{sec:5} contains a brief conclusion.

\section{Mathematical survey on direct conductivity problems and asymptotic expansion formula}\label{sec:2}
Let $\Omega\subset\mathbb{R}^2$ be a smooth, bounded domain that represents a homogeneous medium. We assume that this medium contains a conductivity inclusion $\mathcal{D}$ with small diameter $\eps$ represented as
\[\mathcal{D}=\mz+\eps\mB,\]
where $\mB$ is some fixed bounded domain containing the origin that completely embedded in $\Omega$. Throughout this paper, we denote $\mx,\my$ and $\mz$ are two-dimensional vectors and assume that $\mathcal{D}$ does not touch the boundary $\p\Omega$, i.e., there exists positive constant $h$ such that
\begin{equation}\label{Touching}
  \mbox{dist}(\mathcal{D},\p\Omega)>h.
\end{equation}

Assume that every materials are fully characterized by their electrical conductivity. Let $0<\sigma_0<+\infty$ and $0<\sigma_{\mathcal{D}}<+\infty$ denote the conductivity of the domain $\Omega$ and inclusion $\mathcal{D}$, respectively. By using this notation, we can introduce the following piecewise constant conductivity:
\begin{equation}\label{EPST}
\hat{\sigma}(\mx)=\left\{\begin{array}{ccl}
\sigma_0&\mbox{for}&\mx\in\Omega\backslash\overline{\mathcal{D}}\\
\sigma_{\mathcal{D}}&\mbox{for}&\mx\in\mathcal{D}.
\end{array}\right.
\end{equation}
In this paper, for the sake, we assume that $\sigma>\sigma_0$.

Let $u(\mx)$ be the steady state voltage potential in the presence of the inclusion $\mathcal{D}$, that is, the unique solution to
\begin{equation}\label{Forward}
\nabla\cdot\left(\hat{\sigma}(\mx)\nabla u(\mx)\right)=0\quad\mbox{for}\quad \mx\in\Omega
\end{equation}
with the Neumann boundary condition
\begin{equation}\label{Boundary}
\sigma_0\frac{\partial u}{\partial\nu}(\mx)=g(\mx)\quad\mbox{for}\quad \mx\in\partial\Omega
\end{equation}
and with the compatibility condition
\[\int_{\partial\Omega}u(\mx)d S(\mx)=0.\]
Here $\vn(\mx)$ denote the unit outer normal to $\partial\Omega$ at $x$ and the function $g(\mx)$ represents the applied boundary current satisfies the compatibility condition
\[\int_{\partial\Omega}g(\mx)dS(\mx)=0.\]

Let us denote $u_0$ be the background potential induced by the current $g$ in the domain $\Omega$ without $\mathcal{D}$, that is, the unique solution to
\[\sigma_0\Delta u_0(\mx)=0\quad\mbox{for}\quad \mx\in\Omega\]
with the Neumann boundary condition
\[\sigma_0\frac{\partial u_0}{\partial\nu}(\mx)=g(\mx)\quad\mbox{for}\quad \mx\in\partial\Omega\]
and normalization condition to restore uniqueness
\[\int_{\partial\Omega}u_0(\mx)d S(\mx)=0.\]

Then asymptotic expansion formula in the presence of $\mathcal{D}$ can be written as follows:

\begin{thm}[See \cite{AK2,AS}]
Let $\Omega\subset\mathbb{R}^2$ be a bounded, smooth domain. Then for sufficiently small $\eps$, the asymptotic expansion formula due to the presence of small inclusion $\mathcal{D}=\mx+\eps\mB$ can be expressed in terms of the magnitude of diameter $\eps$ as:
\begin{equation}\label{AsymptoticExpansion}
u(\my)-u_0(\my)=-\pi\eps^2\nabla u_0(\mx)\cdot\mathbb{M}(\mx)\cdot\nabla\mathcal{N}(\mx,\my)+O(\eps^3)
\end{equation}
for each $\my$ on $\partial\Omega$. Here the remaining term $O(\eps^3)$ is uniform in $\mx\in\mathcal{D}$, $\mathcal{N}(\mx,\my)$ denotes the Neumann function for the domain $\Omega$
\begin{equation}\label{Neumannfunction}
\left\{\begin{array}{l}
\sigma_0\Delta \mathcal{N}(\mx,\my)=-\delta(\mx,\my)\quad\mbox{in}\quad\Omega\\
\noalign{\medskip}\displaystyle\sigma_0\frac{\partial\mathcal{N}(\mx,\my)}{\partial\vn(\my)}=-\frac{1}{\abs{\partial\Omega}},\quad \int_{\partial\Omega}\mathcal{N}(\my,\mz)dS(\my)=0\quad\mbox{on}\quad\partial\Omega,
\end{array}\right.
\end{equation}
and $\mathbb{M}(\mx)$ is the symmetric matrix associated with the inclusion $\mathcal{D}$ and the conductivities $\sigma$ and $\sigma_0$. Specially, if $\mathcal{D}$ is a ball with radius $\eps$, $\mathbb{M}(\mx)$ is given by
\begin{equation}\label{MatrixM}
\mathbb{M}(\mx)=2\frac{\sigma_{\mathcal{D}}-\sigma_0}{\sigma_{\mathcal{D}}+\sigma_0}\left(
                                                            \begin{array}{cc}
                                                              1 & 0 \\
                                                              0 & 1 \\
                                                            \end{array}
                                                          \right).
\end{equation}
\end{thm}

\section{Topological derivative and its structure}\label{sec:3}
At this stage, we derive the topological derivative in order to establish an imaging algorithm of small conductivity inclusion $\Sigma$. For a more detailed discussion about the topological derivative, we recommend research articles \cite{AKLP,AM,B1,CR,P2,P5,P6,P7,SZ}.

The topological derivative measures the influence of creating a small ball $\Sigma$ (or crack, etc.) with small radius $r$ at a certain point $\mz$ inside the domain $\Omega$ -- we denote $\Omega\backslash\Sigma$ as the such domain. If we assume that the boundaries $\partial\Omega$ and $\partial\Sigma$ are sufficiently smooth, by considering a cost functional $\mathbb{D}(\Omega)$, the topological derivative $d_T\mathbb{D}(\mz)$ is defined as
\begin{equation}\label{TopDerivative}
  d_T\mathbb{D}(\mz)=\lim_{r\to0+}\frac{\mathbb{D}(\Omega\backslash\Sigma)-\mathbb{D}(\Omega)}{\rho(r)},
\end{equation}
where $\rho(r)\longrightarrow0+$ as $r\longrightarrow0+$. Notice that the function $\rho(r)$ is mainly defined by geometrical factors of the created shape (here, $\Sigma$). With this definition (\ref{TopDerivative}), we have the asymptotic expansion:
\[\mathbb{D}(\Omega\backslash\Sigma)=\mathbb{D}(\Omega)+\rho(r)d_T\mathbb{D}(\mz)+o(\rho(r)).\]

Suppose that $\Omega$ contains a thin inclusion $\mathcal{D}$, $g^{(l)}(\mx)$, $l=1,2,\cdots,L$, be $L$ given functions denotes the boundary condition on $\partial\Omega$ and $u^{(l)}_{\mathcal{D}}(\mx)$ is the solution to the problem in the presence of $\mathcal{D}$:
\begin{equation}\label{ForwardThin}
\left\{\begin{array}{rcl}
\nabla\cdot\left(\hat{\sigma}(\mx)\nabla u^{(l)}_{\mathcal{D}}(\mx)\right)=0&\mbox{in}&\Omega\\
\noalign{\medskip}\displaystyle\sigma_0\frac{\partial u^{(l)}_{\mathcal{D}}}{\partial\nu}(\mx)=g^{(l)}(\mx)&\mbox{on}&\partial\Omega.
\end{array}\right.
\end{equation}
With this, construct $u^{(l)}_0(\mx)$ as the solutions to the following problem in the absence of inclusion:
\begin{equation}\label{ForwardAbsence}
\left\{\begin{array}{rcl}
\sigma_0\triangle u^{(l)}_0(\mx)=0&\mbox{in}&\Omega\\
\noalign{\medskip}\displaystyle\sigma_0\frac{\partial u^{(l)}_0}{\partial\nu}(\mx)=g^{(l)}(\mx)&\mbox{on}&\partial\Omega.
\end{array}\right.
\end{equation}
Let us define following discrepancy function:
\begin{equation}\label{discrepancy}
\mathbb{D}(\Omega):=\frac{1}{2}||u_{\mathcal{D}}(\mx)-u_0(\mx)||_{L^2(\p\Omega)}^2=\frac{1}{2}\int_{\partial\Omega}\abs{u_{\mathcal{D}}(\mx)-u_0(\mx)}^2dS(\mx).
\end{equation}

In order to derive the topological derivative $d_T\mathbb{D}(\mz)$, let us create $\Sigma$ inside the domain $\Omega$ and denote ${u}_{\Sigma}^{(l)}(\mx)$ be the solutions of the following problem in the presence of $\Sigma$:
\[\left\{\begin{array}{rcl}
\nabla\cdot\left(\sigma(\mx)\nabla u_{\Sigma}^{(l)}(\mx)\right)=0&\mbox{in}&\Omega,\\
\noalign{\medskip}\displaystyle\sigma_0\frac{\partial u_{\Sigma}^{(l)}}{\partial\nu}(\mx)=g^{(l)}(\mx)&\mbox{on}&\partial\Omega.
\end{array}\right.\]
where a piecewise constant $\sigma(\mx)$ can be defined similarly with (\ref{EPST}). With this, calculation of the topological derivative can be carried out as follows:

\begin{thm}The topological derivative $d_T\mathbb{D}(\mz)$ of the discrepancy function $\mathbb{D}(\Omega)$ of (\ref{discrepancy}) can be written as follows:
\begin{equation}\label{Disfunc}
\mathbb{D}(\Omega\backslash\Sigma)=\mathbb{D}(\Omega)+\rho(r)d_T\mathbb{D}(\mz)+o(r^2),
\end{equation}
where
\begin{equation}\label{functionRho}
  \rho(r)=\frac{2(\sigma-\sigma_0)}{\sigma+\sigma_0}\pi r^2
\end{equation}
and
\begin{equation}\label{TopSigma}
  d_T\mathbb{D}(\mz)=-\mathrm{Re}\left(\sum_{l=1}^{L}\nabla v^{(l)}_0(\mz)\cdot\overline{\nabla u^{(l)}_0(\mz)}\right).
\end{equation}
Here, $\mathrm{Re}(\ma)$ and $\overline{\ma}$ denotes the real-part and complex conjugate of $\ma$, respectively. Adjoint state $v_0^{(l)}(\mx)$ is defined as the solution to
\begin{equation}\label{ForwardAdjoint}
\left\{\begin{array}{rcl}
\sigma_0\triangle v^{(l)}_0(\mx)=0&\mbox{in}&\Omega\\
\noalign{\medskip}\displaystyle\sigma_0\frac{\partial v^{(l)}_0}{\partial\nu}(\mx)=u_{\mathcal{D}}^{(l)}(\mx)-u_0^{(l)}(\mx)&\mbox{on}&\partial\Omega.
\end{array}\right.
\end{equation}
\end{thm}
\begin{proof}Let us apply equation (\ref{AsymptoticExpansion}) to (\ref{Disfunc}) then we can compute $\mathbb{D}(\Omega\backslash\Sigma)$ as follows:
\begin{align*}
\mathbb{D}(\Omega\backslash\Sigma)&=\frac{1}{2}\sum_{l=1}^{L}\int_{\partial\Omega}\abs{u_{\mathcal{D}}^{(l)}(\mx)-u_{\Sigma}^{(l)}(\mx)}^2dS(\mx)\\
&=\frac{1}{2}\sum_{l=1}^{L}\int_{\partial\Omega}\abs{u_{\mathcal{D}}^{(l)}(\mx)-u_0^{(l)}(\mx)}^2dS(\mx)\\
&\phantom{=}+\sum_{l=1}^{L}\int_{\partial\Omega}\bigg(u_{\mathcal{D}}^{(l)}(\mx)-u_0^{(l)}(\mx)\bigg)\bigg(\overline{u_{0}^{(l)}(\mx)-u_{\Sigma}^{(l)}(\mx)}\bigg)dS(\mx)+o(r^2)\\
&=\mathbb{D}(\Omega)+\mathbb{D}_\Sigma(\mz)+o(r^2).
\end{align*}
where
\[\mathbb{D}_\Sigma(\mz)=\sum_{l=1}^{L}\int_{\partial\Omega}\bigg(u_{\mathcal{D}}^{(l)}(\mx)-u_0^{(l)}(\mx)\bigg)\bigg(\overline{u_{0}^{(l)}(\mx)-u_{\Sigma}^{(l)}(\mx)}\bigg)dS(\mx).\]

By applying asymptotic expansion formula (\ref{AsymptoticExpansion}) and equations (\ref{Neumannfunction}) and (\ref{ForwardAdjoint}), $\mathbb{D}_\Sigma(\mz)$ can be written:
\begin{align*}
  \mathbb{D}_\Sigma(\mz)=&\sum_{l=1}^{L}\int_{\partial\Omega}\bigg(u_0^{(l)}(\mx)-u_{\mathcal{D}}^{(l)}(\mx)\bigg)\bigg(\overline{u_{\Sigma}^{(l)}(\mx)-u_{0}^{(l)}(\mx)}\bigg)dS(\mx)\\
  =&\sum_{l=1}^{L}\int_{\partial\Omega}\sigma_0\frac{\partial v_0^{(l)}}{\partial\nu}(\mx)\bigg(\overline{u_{\Sigma}^{(l)}(\mx)-u_{0}^{(l)}(\mx)}\bigg)dS(\mx)\\
  =&\pi r^2\sum_{l=1}^{L}\int_{\partial\Omega}\sigma_0\frac{\partial v_0^{(l)}}{\partial\nu}(\mx)\bigg(\overline{\nabla u_{0}^{(l)}(\mz)\cdot\mathbb{M}(\mz)\cdot\nabla\mathcal{N}(\mx,\mz)}+o(r^2)\bigg)dS(\mx)\\
  =&\pi r^2\sum_{l=1}^{L}\int_{\Omega}\sigma_0\triangle v_0^{(l)}(\mx)\bigg(\overline{\nabla u_{0}^{(l)}(\mz)\cdot\mathbb{M}(\mz)\cdot\nabla\mathcal{N}(\mx,\mz)}\bigg)dx\\
  &\pi r^2\sum_{l=1}^{L}\int_{\Omega}\nabla v_0^{(l)}(\mx)\cdot\bigg(\overline{\nabla u_{0}^{(l)}(\mz)\cdot\mathbb{M}(\mz)\cdot\sigma_0\triangle\mathcal{N}(\mx,\mz)}\bigg)dx\\
  =&-\pi r^2\sum_{l=1}^{L}\int_{\Omega}\nabla v_0^{(l)}(\mx)\cdot\bigg(\overline{\nabla u_{0}^{(l)}(\mz)\cdot\mathbb{M}(\mz)\delta(\mx,\mz)}\bigg)dx\\
  =&-\pi r^2\sum_{l=1}^{L}\nabla v_0^{(l)}(\mz)\cdot\mathbb{M}(\mz)\cdot\overline{\nabla u_0^{(l)}(\mz)}.
\end{align*}
for $l=1,2,\cdots,L$. Therefore, by (\ref{MatrixM}),
\begin{equation}\label{JSigma}
\mathbb{D}_\Sigma(\mz)=-2\frac{\sigma-\sigma_0}{\sigma+\sigma_0}\pi r^2\sum_{l=1}^{L}\nabla v_0^{(l)}(\mz)\cdot\overline{\nabla u_0^{(l)}(\mz)}.
\end{equation}
Finally, by taking real part of (\ref{JSigma}), we can obtain equations (\ref{functionRho}) and (\ref{TopSigma}).
\end{proof}

\section{Structure of topological derivative}\label{sec:4}
From now on, we identify the structure of topological derivative (\ref{TopSigma}) and discuss certain properties.

Since $v_0^{(l)}$ satisfies (\ref{ForwardAdjoint}), following relations hold for $\mz\in\Omega$,
\begin{align*}
  v_0^{(l)}(\mz)&=\int_{\partial\Omega}\sigma_0\frac{\partial v_0^{(l)}(\my)}{\partial\vn(\my)}\mathcal{N}(\mz,\my)dS(\my)\\
  \nabla_{\mz} v_0^{(l)}(\mz)&=\int_{\partial\Omega}\sigma_0\frac{\partial v_0^{(l)}(\my)}{\partial\vn(\my)}\nabla_{\mz}\mathcal{N}(\mz,\my)dS(\my).
\end{align*}

With them, applying boundary condition of (\ref{ForwardAdjoint}) and asymptotic expansion formula (\ref{AsymptoticExpansion}) yields
\begin{align}
\begin{aligned}\label{TopA}
  \nabla_{\mz} v_0^{(l)}(\mz)&=\int_{\partial\Omega}\sigma_0\bigg(u_{\mathcal{D}}^{(l)}(\my)-u_0^{(l)}(\my)\bigg)\nabla_{\mz}\mathcal{N}(\mz,\my)dS(\my)\\
  &=\pi r^2\int_{\partial\Omega}\bigg(\nabla u_0^{(l)}(\mx)\cdot\mathbb{M}(\mx)\cdot\nabla_{\mx}\mathcal{N}(\mx,\my)\bigg)\nabla_{\mz}\mathcal{N}(\mz,\my)dS(\my).
\end{aligned}
\end{align}
Therefore, (\ref{TopSigma}) can be written as follows:
\begin{align*}
  d_T\mathbb{D}(\mz)&=\pi r^2\mbox{Re}\left\{\left[\int_{\partial\Omega}\left(\nabla u_0^{(l)}(\mx)\cdot\mathbb{M}(\mx)\cdot\nabla_{\mx}\mathcal{N}(\mx,\my)\right)\nabla_{\mz}\mathcal{N}(\mz,\my)dS(\my)\right]\cdot\overline{\nabla u_0^{(l)}(\mz)}\right\}\\
  &=\pi r^2\mbox{Re}\left\{\nabla u_0^{(l)}(\mx)\cdot\mathbb{M}(\mx)\cdot\overline{\nabla u_0^{(l)}(\mz)}\left(\int_{\partial\Omega}\nabla_{\mx}\mathcal{N}(\mx,\my)\nabla_{\mz}\mathcal{N}(\mz,\my)dS(\my)\right)\right\}.
\end{align*}

From the fact that in the two-dimensional space, Neumann function can be decomposed with the singular and regular parts:
\begin{equation}\label{Neumanndecomposition}
  \mathcal{N}(\mx,\my)=\mathcal{S}(\mx,\my)+\mathcal{R}(\mx,\my)=-\frac{1}{2\pi}\ln|\mx-\my|+\mathcal{R}(\mx,\my),
\end{equation}
where $\mathcal{R}(\mx,\my)\in W^{\frac{3}{2},2}(\Omega)$ for any $\my\in\Omega$ and solves
\[\left\{\begin{array}{rcl}
\Delta \mathcal{R}(\mx,\my)=0&\mbox{for}&\mx\in\Omega\\
\noalign{\medskip}\displaystyle\frac{\partial\mathcal{R}(\mx,\my)}{\partial\vn(\mx)}= -\frac{1}{\abs{\partial\Omega}}+\frac{1}{\pi}\frac{\mx-\my}{|\mx-\my|^2}\cdot\vn(\mx)&\mbox{for}&\mx\in\partial\Omega.
\end{array}\right.\]

Particularly, when the domain is a ball with radius $R$ centered at origin, Neumann function is written by (see \cite{AK1})
\[\mathcal{N}(\mx,\my)=-\frac{1}{2\pi}\ln|\mx-\my|-\frac{1}{2\pi}\ln\abs{\frac{R}{|\mx|}\mx-\frac{|\mx|}{R}\my}+\frac{\ln R}{\pi}.\]

With this decomposition, we consider the following value
\begin{align*}
  \int_{\partial\Omega}&\nabla_{\mx}\mathcal{N}(\mx,\my)\nabla_{\mz}\mathcal{N}(\mz,\my)dS(\my)\\
  =&\int_{\partial\Omega}\bigg(-\frac{1}{2\pi}\frac{\mx-\my}{|\mx-\my|}+\nabla_{\mx}\mathcal{R}(\mx,\my)\bigg)\bigg(-\frac{1}{2\pi}\frac{\mz-\my}{|\mz-\my|}+\nabla_{\mz}\mathcal{R}(\mz,\my)\bigg)dS(\my)\\
  =&\frac{1}{4\pi^2}\underbrace{\int_{\partial\Omega}\bigg(\frac{\mx-\my}{|\mx-\my|}\cdot\frac{\mz-\my}{|\mz-\my|}\bigg)dS(\my)}_{:=\mathcal{T}_1} -\frac{1}{2\pi}\underbrace{\int_{\partial\Omega}\bigg(\frac{\mx-\my}{|\mx-\my|}\cdot\nabla_{\mz}\mathcal{R}(\mz,\my)\bigg)dS(\my)}_{:=\mathcal{T}_2} \\
  &-\frac{1}{2\pi}\underbrace{\int_{\partial\Omega}\bigg(\frac{\mz-\my}{|\mz-\my|}\cdot\nabla_{\mx}\mathcal{R}(\mx,\my)\bigg)dS(\my)}_{:=\mathcal{T}_3}
  +\underbrace{\int_{\partial\Omega}\nabla_{\mx}\mathcal{R}(\mx,\my)\cdot\nabla_{\mz}\mathcal{R}(\mz,\my)dS(\my)}_{:=\mathcal{T}_4} .
\end{align*}

Note that since $\mathcal{R}(\mx,\my)\in W^{\frac{3}{2},2}(\Omega)$, there is no blowup of $\nabla_{\mx}\mathcal{R}(\mx,\my)$ so that there exists a constant $C$ such that
\[|\nabla_{\mx}\mathcal{R}(\mx,\my)|\leq C.\]
Hence applying H{\"o}lder's inequality yields
\begin{align*}
  |\mathcal{T}_4|&=\left|\int_{\partial\Omega}\nabla_{\mx}\mathcal{R}(\mx,\my)\cdot\nabla_{\mz}\mathcal{R}(\mz,\my)dS(\my)\right|\leq \int_{\partial\Omega}|\nabla_{\mx}\mathcal{R}(\mx,\my)\cdot\nabla_{\mz}\mathcal{R}(\mz,\my)|dS(\my)\\
  &\leq C^2\mbox{length}(\partial\Omega),
\end{align*}
where $\mbox{length}(\p\Omega)$ means the length of $\partial\Omega$. Moreover, since $\mx\in\mathcal{D}$ and $\my\in\partial\Omega$, $\mx\ne\my$. Hence
\begin{align*}
  |\mathcal{T}_2|&=\left|\int_{\partial\Omega}\bigg(\frac{\mx-\my}{|\mx-\my|}\cdot\nabla_{\mz}\mathcal{R}(\mz,\my)\bigg)dS(\my)\right|\leq \int_{\partial\Omega}\bigg|\frac{\mx-\my}{|\mx-\my|}\cdot\nabla_{\mz}\mathcal{R}(\mz,\my)\bigg|dS(\my)\\
  &\leq C\mbox{length}(\partial\Omega).
\end{align*}
Since $\mz\in\Omega$ and $\my\in\partial\Omega$, we must consider the singularity in $\mathcal{T}_1$ and $\mathcal{T}_3$. For this purpose, we generate a ball $\Omega_B$ of small radius $\delta$ centered at $\my$ and separate $\partial\Omega$ into $\partial\Omega_S=\Omega\cap\Omega_B$ and $\partial\Omega_R=\partial\Omega\backslash\partial\Omega_S$, refer to Figure \ref{fig:2}. Then
\begin{align*}
  |\mathcal{T}_3|&=\left|\int_{\partial\Omega_S}\left(\frac{\mz-\my}{|\mz-\my|}\cdot\nabla_{\mx}\mathcal{R}(\mx,\my)\right)dS(\my)
  +\int_{\partial\Omega_R}\left(\frac{\mz-\my}{|\mz-\my|}\cdot\nabla_{\mx}\mathcal{R}(\mx,\my)\right)dS(\my)\right|\\
  &\leq\lim_{\delta\to0+}\left(\int_{\partial\Omega_S}\left|\frac{\mz-\my}{|\mz-\my|}\cdot\nabla_{\mx}\mathcal{R}(\mx,\my)\right|dS(\my) +\int_{\partial\Omega_R}\left|\frac{\mz-\my}{|\mz-\my|}\cdot\nabla_{\mx}\mathcal{R}(\mx,\my)\right|dS(\my)\right)\\
  &\leq\lim_{\delta\to0+}\bigg(C\frac{\delta}{\delta}\mbox{length}(\partial\Omega_S)+C\mbox{length}(\partial\Omega_R)\bigg) =C\mbox{length}(\partial\Omega)
\end{align*}
and
\begin{align*}
  |\mathcal{T}_1|&=\left|\int_{\partial\Omega}\left(\frac{\mx-\my}{|\mx-\my|}\cdot\frac{\mz-\my}{|\mz-\my|}\right)dS(\my)\right| \leq\int_{\partial\Omega}\left|\frac{\mx-\my}{|\mx-\my|}\right|\left|\frac{\mz-\my}{|\mz-\my|}\right|dS(\my)\\
  &\leq\lim_{\delta\to0+}\left(\int_{\partial\Omega_S}\abs{\frac{\mz-\my}{|\mz-\my|}}dS(\my) +\int_{\partial\Omega_R}\abs{\frac{\mz-\my}{|\mz-\my|}}dS(\my)\right)\\
  &\leq\lim_{\delta\to0+}\bigg(\frac{\delta}{\delta}\mbox{length}(\partial\Omega_S)+\mbox{length}(\partial\Omega_R)\bigg) =\mbox{length}(\partial\Omega).
\end{align*}
Therefore, we can conclude that since
\begin{align*}
  &\bigg|\int_{\partial\Omega}\nabla_{\mx}\mathcal{N}(\mx,\my)\nabla_{\mz}\mathcal{N}(\mz,\my)dS(\my)\bigg|\leq|\mathcal{T}_1|+|\mathcal{T}_2| +|\mathcal{T}_3|+|\mathcal{T}_4|\\
  &\leq\frac{1}{4\pi^2}\mbox{length}(\partial\Omega)+\frac{C}{2\pi}\mbox{length}(\partial\Omega)+\frac{C}{2\pi}\mbox{length}(\partial\Omega) +C^2\mbox{length}(\partial\Omega)<\infty,
\end{align*}
there is no blowup. Therefore, the structure of normalized topological derivative (\ref{TopSigma}) will be of the form:
\[\frac{d_T\mathbb{D}(\mz)}{\max|d_T\mathbb{D}(\mz)|}\approx\sum_{l=1}^{L}\nabla u_0^{(l)}(\mx)\cdot\mathbb{M}(\mx)\cdot\overline{\nabla u_0^{(l)}(\mz)}.\]

\begin{figure}[!ht]
\centering
\includegraphics[width=0.3\textwidth]{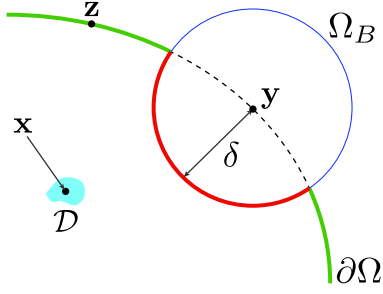}
\caption{Sketch of $\partial\Omega_S$ (red line) and $\partial\Omega_R$ (green line).}\label{fig:2}
\end{figure}

\subsection{Limitation of topological derivative to EIT problem}
In practice, most of EIT systems generally apply constant current sources in several direction so that the boundary condition (\ref{Boundary}) can be written with for a constant vector $\ma_l$, $l=1,2,\cdots,L$,
\begin{equation}\label{BoundaryEIT}
  g^{(l)}(\mx)=\sigma_0\frac{\partial u}{\partial\nu}(\mx)=\ma_l\cdot\vn(\mx)\quad\mbox{for}\quad \mx\in\partial\Omega.
\end{equation}
Then since background potential $u_0^{(l)}$ is
\[u_0^{(l)}(\mx)=\ma_l\cdot\mx-\frac{1}{|\partial\Omega|}\int_{\partial\Omega}\ma_l\cdot\vn(\my)dS(\my),\]
normalized topological derivative (\ref{TopSigma}) becomes
\[\frac{d_T\mathbb{D}(\mz)}{\max|d_T\mathbb{D}(\mz)|}\approx\nabla u_0^{(l)}(\mx)\cdot\mathbb{M}(\mx)\cdot\overline{\nabla u_0^{(l)}(\mz)}=\ma_l\cdot \ma_l=4\frac{\sigma_{\mathcal{D}}-\sigma_0}{\sigma_{\mathcal{D}}+\sigma_0}.\]
Therefore, (\ref{TopSigma}) does not offers any information of $\mx\in\mathcal{D}$. This is a reason that why topological derivative concept cannot be applied to the EIT problem.

\subsection{Alternative boundary condition and corresponding topological derivative}
This is motivated from the original idea in \cite{C}. We would like to mention \cite{AK2,AMV} for its application. Suppose that there is only one inhomogeneity $\mathcal{D}$ exists in $\Omega$. Let us consider the following boundary conditions
\begin{align*}
  g_1^{(l)}(\mx)&=ik(\vt_l+i\vt_l^\bot)\cdot\vn(\mx)\exp\bigg(ik(\vt_l+i\vt_l^\bot)\cdot\mx\bigg)\\
  g_2^{(l)}(\mx)&=ik(\vt_l-i\vt_l^\bot)\cdot\vn(\mx)\exp\bigg(ik(\vt_l-i\vt_l^\bot)\cdot\mx\bigg).
\end{align*}
Then it is easy to observe that
\[u_1^{(l)}(\mx)=\exp\bigg(ik(\vt_l+i\vt_l^\bot)\cdot\mx\bigg)\quad\mbox{and}\quad u_2^{(l)}(\mx)=\exp\bigg(ik(\vt_l-i\vt_l^\bot)\cdot\mx\bigg).\]
are satisfy (\ref{ForwardAbsence}) when the boundary conditions are $g_1^{(l)}(\mx)$ and $g_2^{(l)}(\mx)$, respectively. Here, $k$ is a positive real number, $\vt_l$ is an arbitrary vector on the two-dimensional unit circle $\mathbb{S}^1$ and $\vt_l^\bot\in\mathbb{S}^1$ is orthogonal to $\vt$ with $|\vt_l|=|\vt_l^\bot|=1$. In this paper, we set $\vt^\bot=[\theta_1,\theta_2]^\bot=[\theta_2,-\theta_1]$.

We introduce an alternative topological derivatives with respect to $g_1^{(l)}(\mx)$ and $g_2^{(l)}(\mx)$ as
\begin{equation}\label{ATD}
  \mathbb{D}_\mathrm{A}(\mz;k):=\sum_{l=1}^{L}\bigg(\nabla v_1^{(l)}(\mz)\cdot\overline{\nabla u_1^{(l)}(\mz)}\bigg)\bigg(\nabla v_2^{(l)}(\mz)\cdot\overline{\nabla u_2^{(l)}(\mz)}\bigg),
\end{equation}
where $v_1^{(l)}(\mz)$ and $v_2^{(l)}(\mz)$ satisfy (\ref{ForwardAdjoint}) with respect to $g_1^{(l)}(\mx)$ and $g_2^{(l)}(\mx)$, respectively. Then, if $L$ is sufficiently large enough, we can obtain
\begin{align*}
  \mathbb{D}_\mathrm{A}(\mz;k)\approx&\sum_{l=1}^{L}\bigg(\nabla u_1^{(l)}(\mx)\cdot\mathbb{M}(\mx)\cdot\overline{\nabla u_1^{(l)}(\mz)}\bigg)\bigg(\nabla u_2^{(l)}(\mx)\cdot\mathbb{M}(\mx)\cdot\overline{\nabla u_2^{(l)}(\mz)}\bigg)\\
  =&\int_{\mathbb{S}^1}\left(4k\frac{\sigma_{\mathcal{D}}-\sigma_0}{\sigma_{\mathcal{D}}+\sigma_0} \exp\bigg(ik(\vt+i\vt^\bot)\cdot\mx\bigg)\overline{\exp\bigg(ik(\vt+i\vt^\bot)\cdot\mz\bigg)}\right)\\
  &\times\left(4k\frac{\sigma_{\mathcal{D}}-\sigma_0}{\sigma_{\mathcal{D}}+\sigma_0} \exp\bigg(ik(\vt-i\vt^\bot)\cdot\mx\bigg)\overline{\exp\bigg(ik(\vt-i\vt^\bot)\cdot\mz\bigg)}\right)d\vt\\
  =&\left(4k\frac{\sigma_{\mathcal{D}}-\sigma_0}{\sigma_{\mathcal{D}}+\sigma_0}\right)^2\int_{\mathbb{S}^1}\exp\bigg(2ik\vt\cdot(\mx-\mz)\bigg)d\vt=16k^2\pi\left(\frac{\sigma_{\mathcal{D}}-\sigma_0}{\sigma_{\mathcal{D}}+\sigma_0}\right)^2J_0(2k|\mx-\mz|).
\end{align*}
Here, following identity used (see \cite{G,P8}); for sufficiently large $L$,
\[\sum_{l=1}^{L}\exp(ik\vt_l\cdot\mx)\approx\int_{\mathbb{S}^1}\exp(ik\vt\cdot\mx)d\vt=2\pi J_0(k\mx),\]
where $J_0(x)$ is the Bessel function of order zero and of the first kind. Therefore, normalized topological derivative will be of the form
\[d_T\mathbb{D}_\mathrm{A}(\mz;k)=\frac{\mathbb{D}_\mathrm{A}(\mz;k)}{\max|\mathbb{D}_\mathrm{A}(\mz;k)|}=J_0(2k|\mx-\mz|).\]
This gives some certain properties of $\mathbb{D}_\mathrm{A}(\mz;k)$ summarized as follows.

\begin{enumerate}
  \item Since $J_0(x)$ reaches its maximum value $1$ at $x=0$, $d_T\mathbb{D}_\mathrm{A}(\mz;k)$ will plot its maximum value at $\mz=\mx\in\mathcal{D}$ so that we can identify the location of $\mathcal{D}$.
  \item Resolution of image is highly depends on the value of $k$ and $L$. Based on the property of $J_0(x)$ (see Figure \ref{PlotBesselFunction}), one can obtain a result with high resolution if $k$ is sufficiently large. In contrast, if the value of $k$ is small, one cannot identify the location of $\mathcal{D}$, refer to Figure \ref{Result1}.
  \item Due to the assumption of (\ref{Touching}), if $\mathcal{D}$ is (nearly) touching the boundary $\p\Omega$, it is very hard to conclude that $d_T\mathbb{D}_\mathrm{A}(\mz;k)$ yields a good image.
  \item Asymptotic expansion formula (\ref{AsymptoticExpansion}) holds for small inhomogeneity in theory. Therefore, when $\mathcal{D}$ is an extended target, further analysis of $d_T\mathbb{D}_\mathrm{A}(\mz;k)$ is required.
\end{enumerate}

\begin{figure}[!ht]
\begin{center}
  \includegraphics[width=0.99\textwidth]{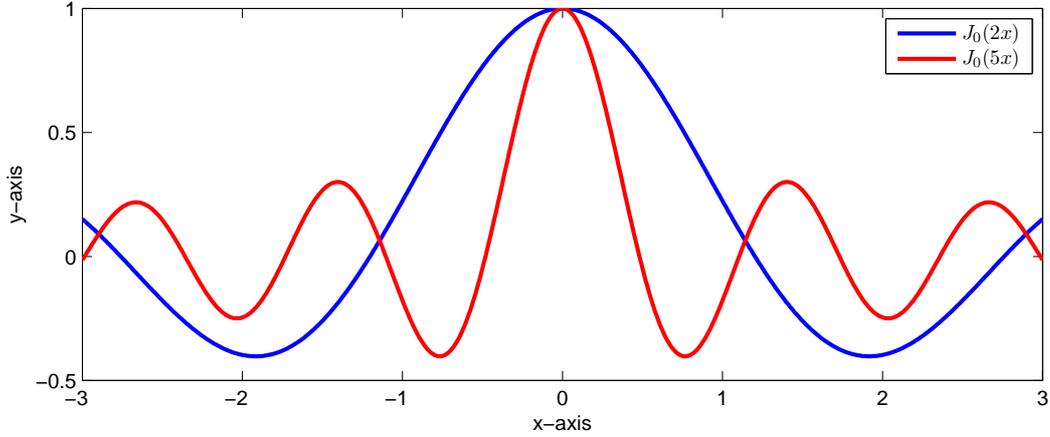}
\end{center}
\caption{\label{PlotBesselFunction}Plot of functions $y=J_0(kx)$ for $k=2$ and $k=5$.}
\end{figure}

\section{Numerical experiments and related discussions}\label{sec:5}
In this section, a numerical result is shown for showing the effectiveness of (\ref{ATD}). For this, we take the domain $\Omega$ to be the unit circle and we insert one inhomogeneity $\mathcal{D}_1$ in the shape of ball with the radius $\eps_1=0.1$, the location $\mx_1=[0.4,0.3]^T$, and the conductivity $\sigma_{\mathcal{D}_1}=5$. $L=16$ incident directions are applied and every forward problems (\ref{Forward}), (\ref{ForwardAbsence}), (\ref{ForwardAdjoint}) are solved via traditional Finite Element Method (FEM) in order to avoid the inverse crime. After the generation of boundary measurement data, a noise is added as follows
\[u_{\mathcal{D},\mathrm{noise}}^{(l)}(\my)=\left(1+\xi\{\mathrm{rand}_1(-1,1)+i\mathrm{rand}_2(-1,1)\}u_{\mathcal{D}}^{(l)}(\my)\right),\]
where $\mathrm{rand}_1(-1,1)$ and $\mathrm{rand}_2(-1,1)$ are arbitrary real values between $-1$ and $1$. Throughout this paper, we take $\xi=0.3$. The values $\nabla v_j^{(l)}(\mz)$ and $\nabla u_j^{(l)}(\mz)$, $j=1,2$, of (\ref{ATD}) are evaluated by the Matlab command \texttt{pdegrad} included in the Partial Differential Equation Toolbox.

From Figure \ref{Result1}, we can observe that the location of $\mathcal{D}_1$ is clearly identified when $k(\geq3)$ is sufficiently large enough, but on the other hand we cannot identity the location when the value of $k(\leq1)$ is small.

\begin{figure}[!ht]
\centering
\includegraphics[width=0.49\textwidth]{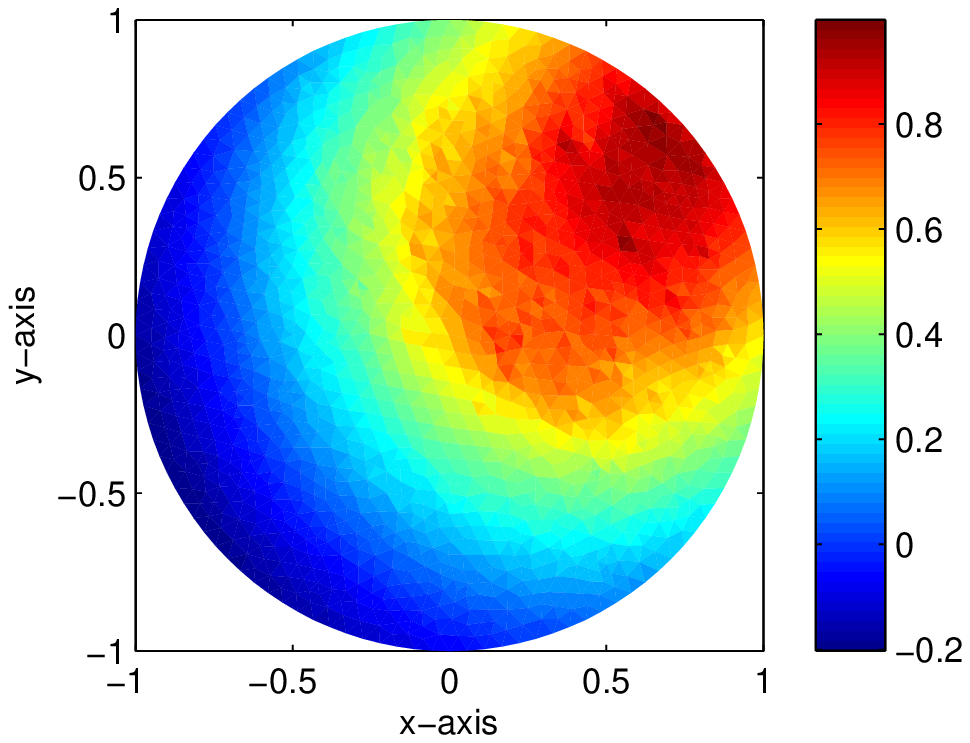}
\includegraphics[width=0.49\textwidth]{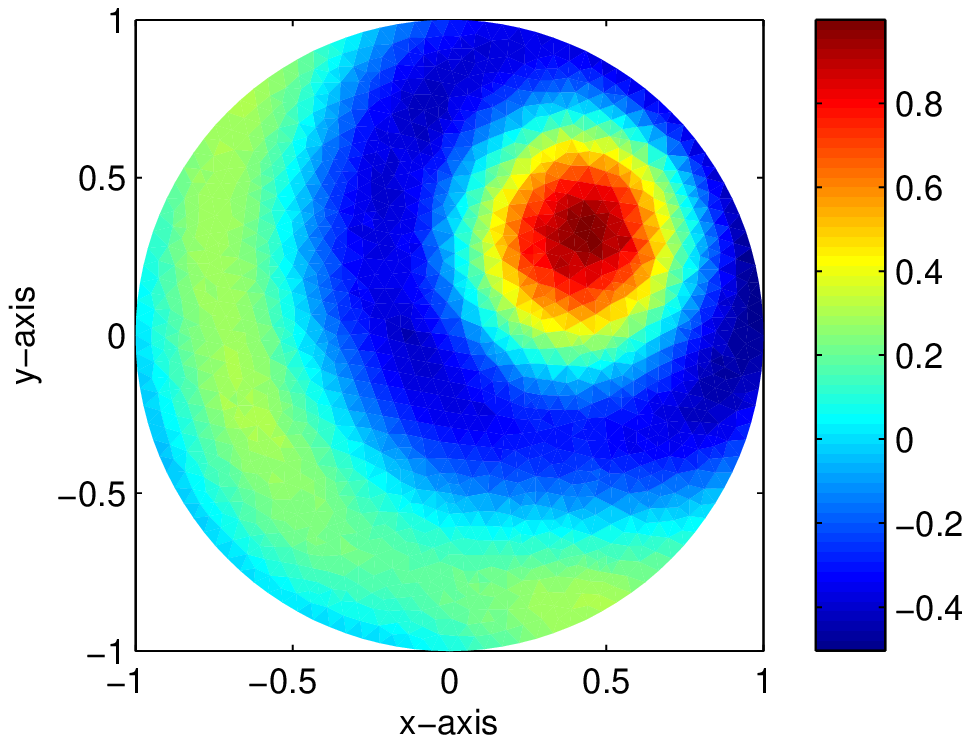}\\
\includegraphics[width=0.49\textwidth]{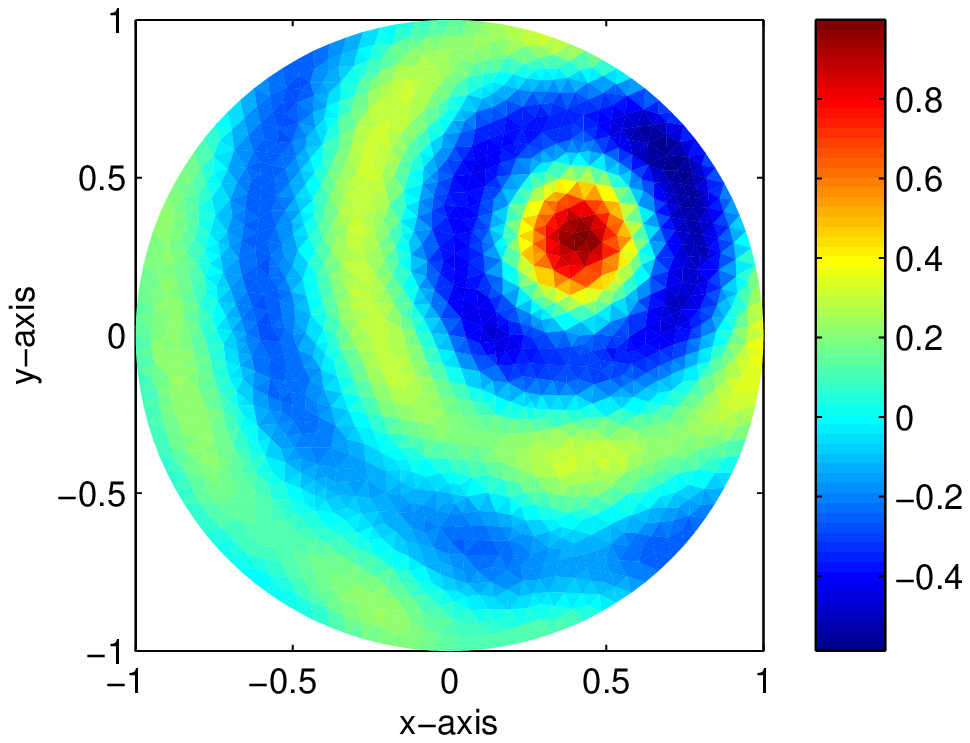}
\includegraphics[width=0.49\textwidth]{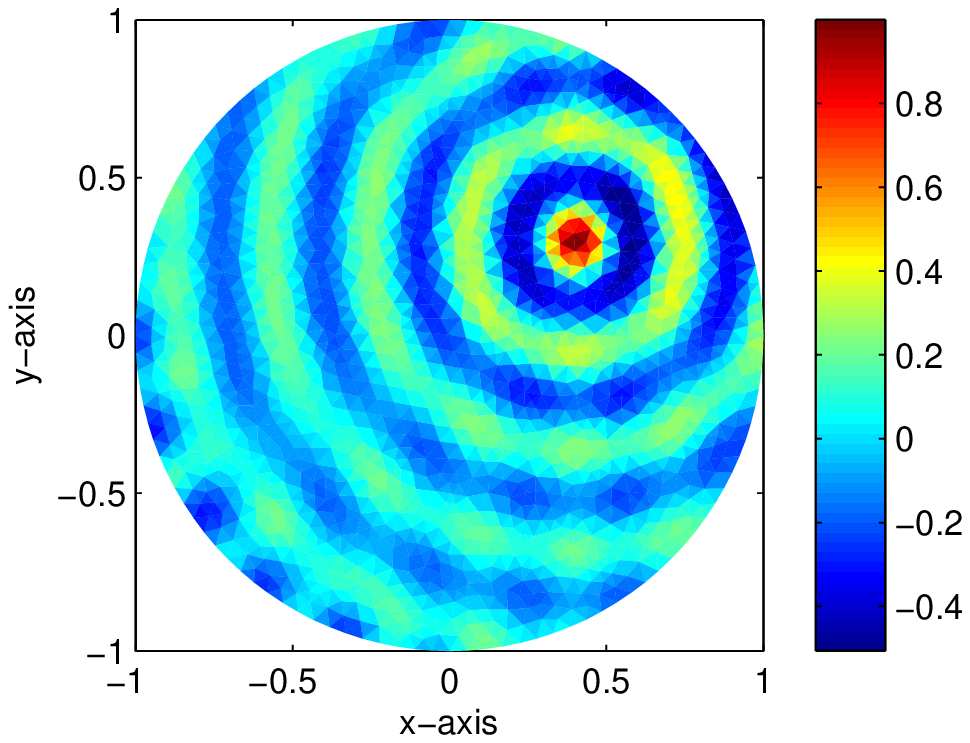}
\caption{\label{Result1}Maps of $d_T\mathbb{D}_\mathrm{A}(\mz;k)$ for $k=1$ (top, left), $k=3$ (top, right), $k=5$ (bottom, left), and $k=10$ (bottom, right).}
\end{figure}

We apply (\ref{ATD}) for finding locations of multiple inhomogeneities. For this, we add another one inhomogeneity $\mathcal{D}_2$ in the shape of ball with the radius $\eps_2=0.1$, the location $\mx_2=[-0.5,-0.2]^T$, and the conductivity $\sigma_{\mathcal{D}_2}=5$. Figure \ref{Result2} shows the maps of $d_T\mathbb{D}_\mathrm{A}(\mz;k)$ with various values of $k$. Unfortunately, in contrast with the imaging of single inhomogeneity, we cannot identify $\mathcal{D}_1$ and $\mathcal{D}_2$.

\begin{figure}[!ht]
\centering
\includegraphics[width=0.49\textwidth]{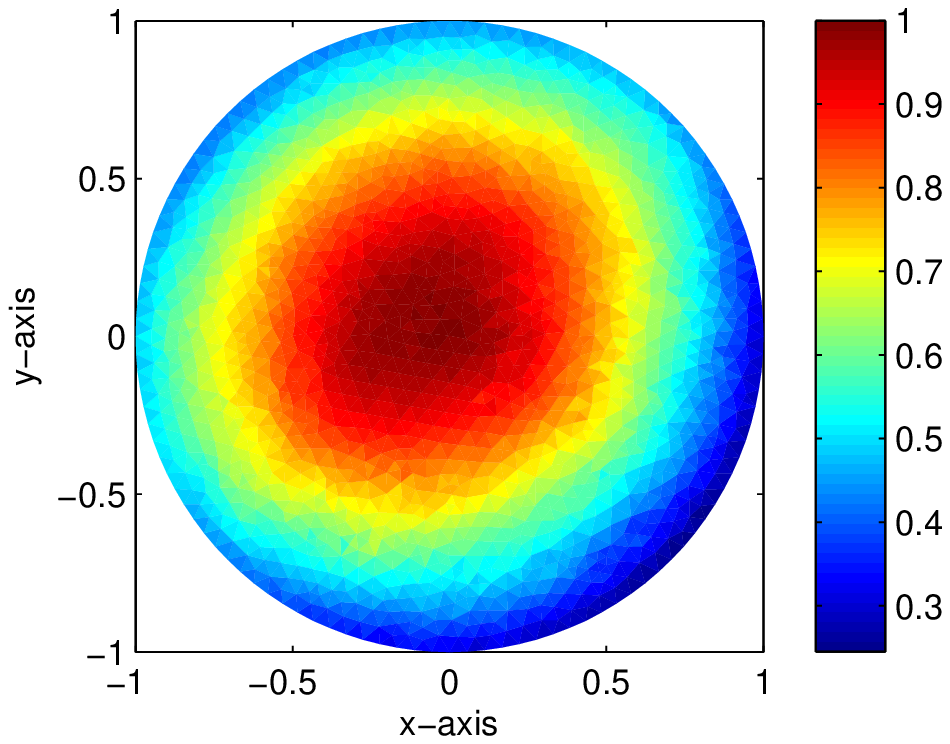}
\includegraphics[width=0.49\textwidth]{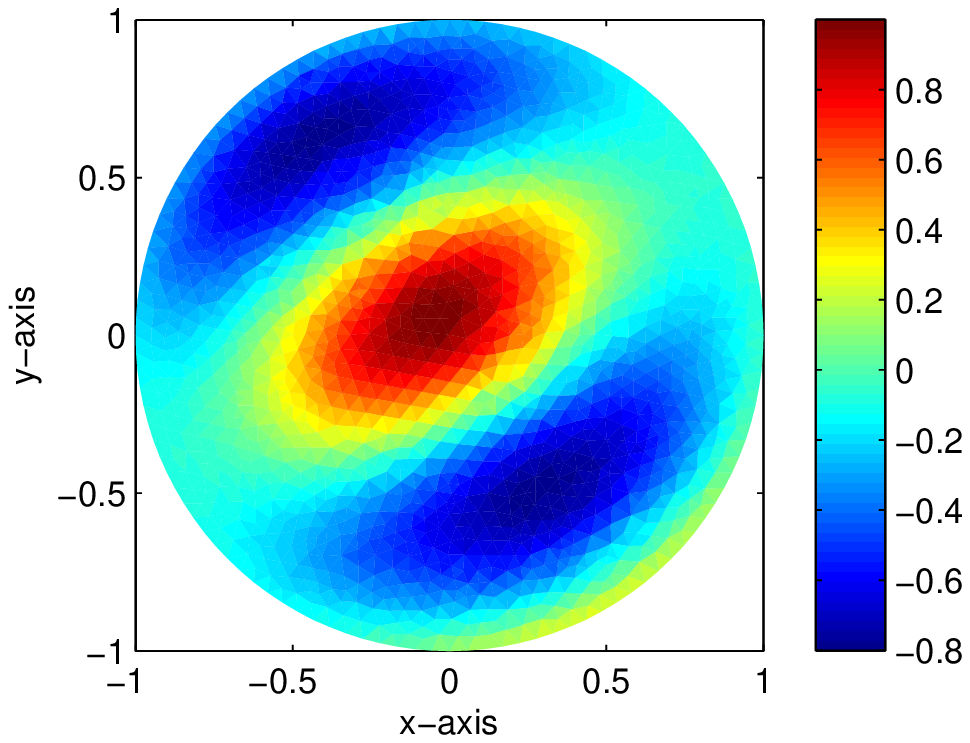}\\
\includegraphics[width=0.49\textwidth]{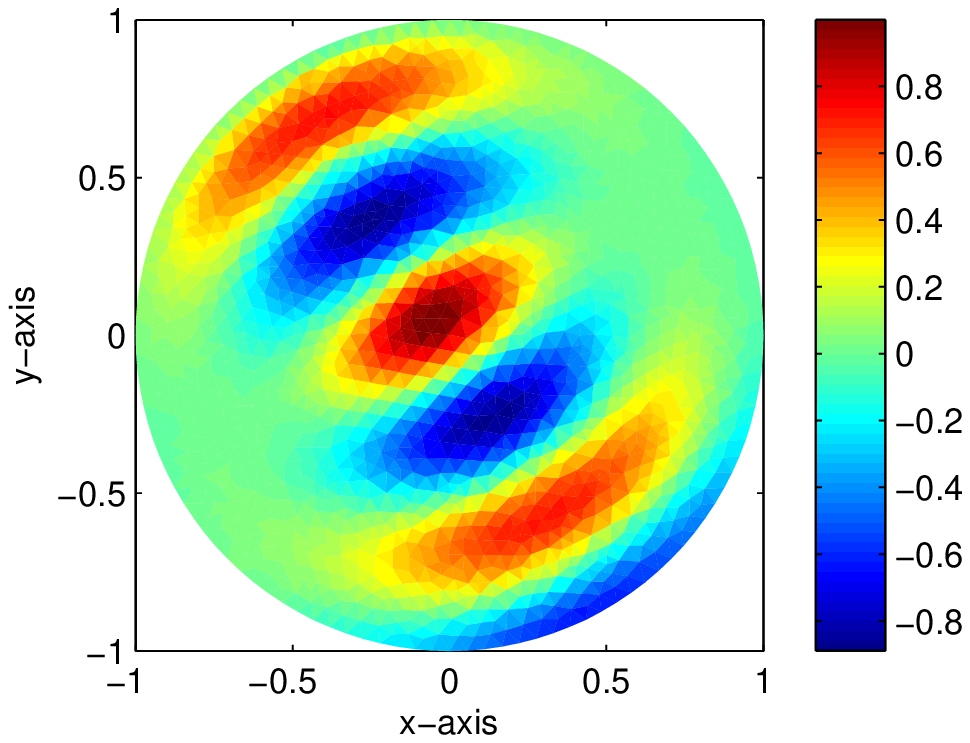}
\includegraphics[width=0.49\textwidth]{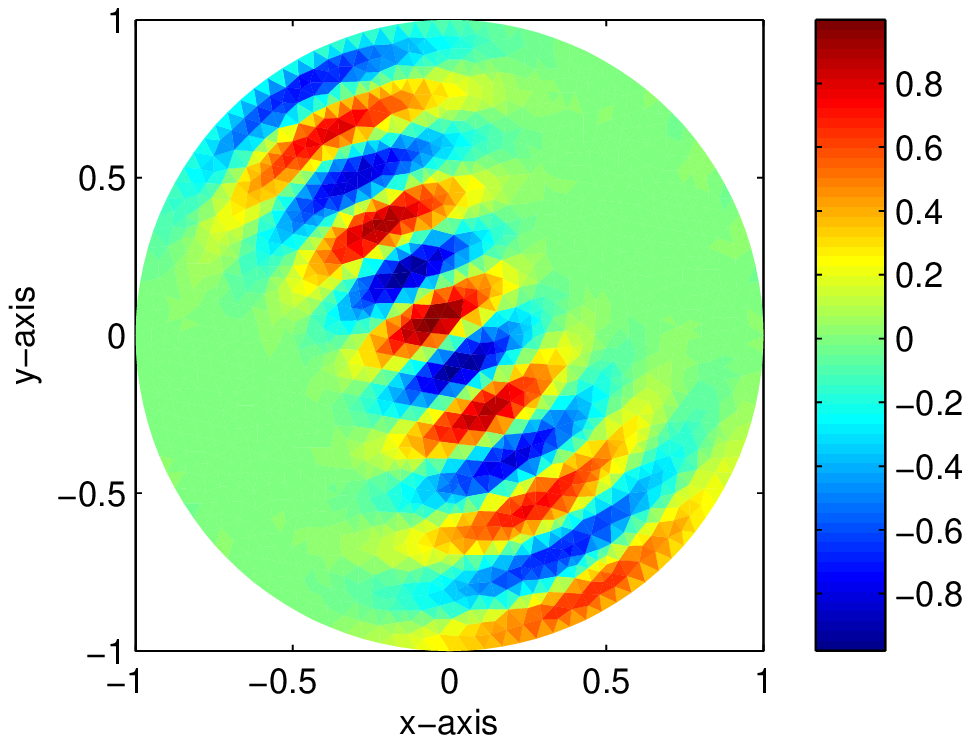}
\caption{\label{Result2}Maps of $d_T\mathbb{D}_\mathrm{A}(\mz;k)$ for $k=1$ (top, left), $k=3$ (top, right), $k=5$ (bottom, left), and $k=10$ (bottom, right).}
\end{figure}

In order to reveal the reason, we reconsider (\ref{ATD}) in the existence of two-different inhomogeneities $\mathcal{D}_1$ and $\mathcal{D}_2$. In this case, by a simple calculation, (\ref{ATD}) becomes
\begin{align*}
  \mathbb{D}_\mathrm{A}(\mz;k):=&\sum_{l=1}^{L}\bigg(\nabla v_1^{(l)}(\mz)\cdot\overline{\nabla u_1^{(l)}(\mz)}\bigg)\bigg(\nabla v_2^{(l)}(\mz)\cdot\overline{\nabla u_2^{(l)}(\mz)}\bigg)\\
  =&\sum_{l=1}^{L}\left\{\sum_{m=1}^{2}\left(\nabla u_1^{(l)}(\mx_m)\cdot\mathbb{M}(\mx_m)\cdot\overline{\nabla u_1^{(l)}(\mz)}\right)\right\}\left\{\sum_{m=1}^{2}\left(\nabla u_2^{(l)}(\mx_m)\cdot\mathbb{M}(\mx_m)\cdot\overline{\nabla u_2^{(l)}(\mz)}\right)\right\}\\
  \approx&2\pi\bigg(\Lambda_1^2J_0(2k|\mx_1-\mz|)+\Lambda_2^2J_0(2k|\mx_2-\mz|)\bigg)+\Phi(\mx_1,\mx_2,\mz,\vt;k),
\end{align*}
where
\[\Phi(\mx_1,\mx_2,\mz,\vt;k):=2\Lambda_1\Lambda_2\int_{\mathbb{S}^1}\cos\bigg(k\vt^\bot\cdot(\mx_1-\mx_2)\bigg)\exp\bigg(ik\vt\cdot(\mx_1+\mx_2-2\mz)\bigg)d\vt\]
and for $j=1,2$,
\[\Lambda_j:=4k\frac{\sigma_{\mathcal{D}_j}-\sigma_0}{\sigma_{\mathcal{D}_j}+\sigma_0}.\]

Notice that due to the periodic property of cosine and exponential functions, many artifacts will appear in the map of $\mathbb{D}_\mathrm{A}(\mz;k)$ so that in contrast to the case of single inclusion, $\mathbb{D}_\mathrm{A}(\mz;k)$ will produces poor result, refer to Figure \ref{Result2}.


\section{Conclusion}\label{sec:5}
We have proposed a one-step iterative algorithm based on the topological derivative concept to image the conductivity inclusion with small diameter. This algorithm is based on the the asymptotic formula for steady state voltage potentials in the existence of such inclusion. Then we have performed some numerical simulations and conclude that although traditional topological derivative does not yields an image of inclusion, proposed alternative topological derivative offers very good result. Hence, it can be reconstructed completely upon by an appropriate iterative algorithms \cite{ADIM,AM,CR,DL,PL4,S}.

Unfortunately, proposed method can be applied for imaging of single, small inclusion. Development of algorithm for imaging of multiple inclusions, arbitrary shaped inclusion such as crack-like thin conductivity inclusion or extended one will be a valuable addition to this work. In this paper, only two-dimensional problem have been considered herein, we expect that the suggested strategy, e.g., mathematical treatment of the asymptotic formula, imaging method, etc., could be extended to the three-dimensional problem.

\end{document}